\documentclass[12pt,a4paper]{article}
\usepackage{amsmath,amsfonts,amssymb,amsthm}
\usepackage{enumerate}
\usepackage{url}
\usepackage{multicol}
\usepackage{graphicx}
\usepackage{graphics}
\newtheorem{theorem}{Theorem}[section]
\newtheorem{lemma}{Lemma}[section]

\title{Collapse of n vortices
\footnote{Presented at the XX Polish Fluid Mechanics Conference,
Gliwice, 17-20 September 2012}}
\author{Marek LEWKOWICZ, Henryk KUDELA}
\def\C{\mathbb{ C}}
\def\R{\mathbb{ R}}

\begin{document}
\maketitle
\pagestyle{headings}
\pagenumbering{arabic}
\section{Introduction}
The importance of the point vortices in applications
is due to the dominant role of the coherent vortical structures
in many two-dimensional turbulent flows. 
Point-vortex dynamics is responsible 
for bringing the vortices together 
and in this way it determines 
what kind of a merging process will occur. 
The number of vortices in a flow 
can be quite large; 
this makes the complicated 
dynamical description intractable. 
On the other hand, 
few-vortex systems 
can be investigated in much more detail. 
The evolution of such systems 
has been studied for more than 130 years.

In 1883 Kirchhoff showed that the motion 
can be put into the Hamilton framework. 
Systems of three vortices are integrable 
since they possess enough 
Poisson commuting invariants. 
It was quite surprising that there exist 
three-vortices systems whose evolution 
leads to a collapse to a point in finite time. 
This phenomenon can be treated 
as a change of scale, 
characteristic for turbulent motion. 
A deeper insight into the way 
systems of vortices can collapse 
could thus lead to a better understanding 
of turbulence itself.

Collapsing systems of three vortices 
were first described by Groebli in 1877 \cite{grobli}
and rediscovered by Aref (see \cite{aref2010})
and independently by Novikov and Sedov \cite{novikov} 
around 1979. 
The motion of four vortices 
is no longer integrable in general. 
Nevertheless in 1979 Novikov and Sedov 
gave some special, explicit examples 
of collapsing systems of four and five vortices.

In 1987 O'Neil proved (\cite{oneil}) the existence 
of collapsing systems of $n$ vortices for arbitrary $n$. 
The proof makes use
of a system of algebraic equations
that self-similar collapsing configurations should satisfy.
O'Neil proved that for some circulations 
the set of solutions is a non-empty algebraic curve 
in the configuration space. 
Thus solutions exist, although the proof does not indicate 
how they should be found. 
In fact \cite{oneil} does not provide any examples 
for the interesting case $n\ge6$.

We solve numerically the above mentioned (non-linear) 
algebraic system of equations 
and obtain collapsing configurations 
for many circulations for six or more vortices. 
To the best of our knowledge no such examples 
have appeared in the literature so far. 
A precise description of the algorithm used is given below. 
The solutions we found are taken as initial conditions 
to the differential evolution equation. 
Standard numerical procedures give then trajectories 
whose collapsing property can be directly seen. 
A sample collapsing system of seven vortices is shown in Fig. 3, 4 and 5. 
Note that in order to obtain a system 
whose diameter lessens sufficiently during evolution, 
it is necessary to use extremely high precision 
both when calculating the initial state 
and when solving the evolution equation. 

It may be worth mentioning that although
the O'Neil's results on existence of collapsing configurations
give some moral support in seeking such configurations,
nevertheless our numerical results do not rely on his theorems.
The examples we found prove independently
that collapsing configurations of $n$ vortices do exist,
obviously only for those values of $n$, 
for which the calculations have been performed.
In particular we have examples of collapsing configurations
for such circulations, for which O'Neil's theorems do not work.

\section{Basic notions and facts}
We start with a review of some basic notions and facts,
which will help us put our results in a proper perspective.
For this basic material one can consult 
e.g. \cite{oneil}, \cite{garduno}.

A two-dimensional fluid motion
can be discretized 
by dividing the vorticity field 
into regions $U_i$ and replacing each of them
by a point vortex $z_i$
given a suitable circulation $\Gamma_i$.
The idea can be traced back to
Helmholtz and Kirchhoff.
The equations of motion of the system $z(t)=(z_1,…,z_n)\in\C^n$ 
of $n$ vortices are 
$$\frac{dz_k(t)}{dt}=V_k(z(t))
=\frac{i}{2\pi}
\sum_{l=1,l\neq k}^n
\Gamma_l\ \frac{z_k-z_l}{|z_k-z_l|^2}$$
with nonzero real numbers $\Gamma_l$.
Let's observe that our dynamical system
is a hamiltonian system with respect to the symplectic form
$$\Omega(z,w)
=-Im\sum_{k=1}^n\Gamma_kz_k\bar{w}_k
=\sum_{k=1}^n\Gamma_kdx_kdy_k.$$
This means that
for some hamiltonian total energy function $H$,
the vector field $V$ should be dual to $dH$ 
with respect to $\Omega$.
By definition, we need $i_{V}\Omega=dH$
or equivalently $\Omega(V,U)=UH$ 
(the directional derivative of $H$ along $U$)
for any vector field $U$. 
In fact
$$H=-\frac12\sum_{k<l}\Gamma_k\Gamma_l\log|z_k-z_l|$$
would do.

Let's introduce the size $S(z)\in\R$ 
and the moment of vorticity $M(z)\in\C$ by
$$S(z)=\sum_k\Gamma_k|z_k|^2,
\quad M(z)=\sum_k\Gamma_kz_k.$$
Observe that our hamiltonian system 
is rotationally and translatory invariant
and therefore, except for the hamiltonian $H$ itself,
it has three additional integrals of motion
(that is real functions conserved along the trajectories):
$S$ and the real and imaginary part of $M$.
It follows that the dynamical system of three vortices 
is integrable.

Two systems of vortices $z$ and $w$ 
are said to form the same configuration 
if they are similar, i.e., $w_k=az_k+b$ 
for some complex $a,b$ with $a\neq0$. 
The configuration space for $n$ vortices has dimension $2n-4$
since the group of similarities $z\to az+b$ is four-dimensional.
It is easy to see that $w_k=az_k+b$
implies $V(w)=\frac{1}{\bar{a}}V(z)$
and therefore any trajectory $t\to z(t)$
gives a trajectory 
$t\to az\left(\frac{t}{|a|^2}\right)+b$
with initial conditions $az(0)+b$.
We express this by saying that similar 
systems are dynamically equivalent. 
A system $z(0)$ is said to collapse if for some $t_0>0$
its trajectory $z(t)$ converges to a point.
All systems similar to $z(0)$ also collapse;
thus the collapsing families are at least four-dimensional.
In fact, as O'Neil's work shows, quite often they are five-dimensional.
It is more convenient to speak of collapsing configurations 
instead of collapsing systems:
the families described by O'Neil are one-dimensional 
(form algebraic curves in the configuration space).

A system is called self-similar 
if it remains similar to the initial state during evolution.
This amounts to say that the trajectory $z(t)$
is a fixed point in the configuration space.
One also says its configuration is stationary.
It is believed that any collapsing system must be self-similar.
Therefore all collapsing systems considered in this paper
are assumed to be self-similar.

\begin{lemma}
Let $w$ be a system of vortices and $z(t)$ the trajectory starting at $w$. 
The following conditions are equivalent.
\begin{itemize}
\item (a) The system $w$ is self-similar.
\item (b) $V_k(w)-V_l(w)=\omega(w_k-w_l)$ for some $\omega\in\C$.
\item (c) $V_k(w)=\omega w_k-p$ for some $\omega,p\in\C$.
\item (d) The system $w$ belongs to one of the following five classes.
\begin{itemize}
\item (i) stationary: $V_k(w)=0$ or equivalently $z(t)=w$,
\item (ii) translatory: $V_k(w)=v\neq0$ or equivalently  $z(t)=w+tv$,
\item (iii) rotational: $V_k(w)=i\lambda(w_k-p)$, $\lambda\neq0$,
$z(t)=p+e^{i\lambda t}(w-p)$,
\item (iv) collapsing: $V_k(w)=\omega(w_k-p)$,
$Re(\omega)<0$,
$$z(t)=p+\sqrt{2 Re(\omega)t+1}
e^{i\frac{Im(\omega)}{2 Re(\omega)}\ln(2 Re(\omega)t+1)}(w-p),$$
\item (v) expanding: $V_k(w)=\omega(w_k-p)$,
$Re(\omega)>0$, $z(t)$ given by the same 
formula as for the collapsing system.
\end{itemize}
\end{itemize}
\end{lemma}
Note that a collapsing trajectory is defined
for $t\in(-\infty,t_0)$,
$t_0=\frac{-1}{2 Re(\omega)}>0$,
and that in fact $\lim_{t\to t_0}|z_k(t)-p|=0$.
Note also that $\omega$ in (b), (c) and (d) is the same. 
The zero value of $\omega$ corresponds to (d-i..ii).
A system is collapsing, expanding or rotational (d-iii..v)
iff it satisfies (b) or (c) with $\omega\neq0$.

\section{Necessary conditions for collapsing}
\nopagebreak

In this section we shall specify 
and reformulate algebraic conditions
that collapsing configurations 
and their circulations should satisfy.
These equations will be solved numerically 
in the following sections.

We start with conditions for circulations.
For a circulation $\Gamma=(\Gamma_k)_k$
we define the angular momentum $L$ 
and the total circulation $\sigma$ by
$$L=\sum_{i<j}\ \Gamma_i\Gamma_j,
\quad\sigma=\sum_i\Gamma_i.$$
\begin{lemma}
If a collapsing configuration exists then necessarily
$$L=0,\quad\sigma\neq0.$$
\end{lemma}
\begin{proof}
Suppose $z(t)$ is a trajectory and for some $t$ 
$z(t)$ is similar to $z(0)$ 
with the proportionality factor $c\neq1$.
Thus $\log|z_k(t)-z_l(t)|=\log c|z_k(0)-z_l(0)|
=\log|z_k(0)-z_l(0)|+\log c$
and $H(z(t))=H(z(0))-\frac12L\log c$.
Since the hamiltonian is constant 
along the trajectory and $\log c\neq0$,
$L$ must  be zero.
Furthermore
$$\sigma^2-2L=\left(\sum_i\Gamma_i\right)^2
-\sum_{i\neq j}\Gamma_i\Gamma_j
=\sum_i\Gamma_i^2>0,$$
so that $L=0$ implies $\sigma\neq0$.
\end{proof}
From now on we assume that $L=0$, $\sigma\neq0$.
\begin{lemma}
If $\sigma\neq0$ then any system $z$ 
admits a unique translate $w=z-p$
(which means that $w_k=z_k-p$ for some $p\in\C$), 
such that $M(w)=0$.
\end{lemma}
\begin{proof}
The moment of the translate is
$$M(w)=M(z-p)=\sum\Gamma_k(z_k-p)=M(z)-p\sigma.$$
Thus $p=M(w)/\sigma$.
\end{proof}
It follows that each configuration class
contains a representative satisfying $M=0$
and that representative is unique up to a complex factor.
In other words the configuration space 
is in a natural one-to-one
correspondence with the complex projective space $\C P(W)$
over the kernel $W$ of the linear map $M$.
This will be of practical importance since now
a set of configurations can be specified 
by a system of homogeneous equations,
one of them being $M=0$.
\begin{lemma}
(See \cite{oneil}, Lemma 1.2.4 and Lemma 1.2.7)
Suppose that $L=0$, $\sigma\neq0$, $M(w)=0$.
Then the following conditions are equivalent
\begin{itemize}
\item (a) $w$ is self-similar and expanding, 
collapsing or rotational.
\item (b) $V_k(w)-V_l(w)=\omega(w_k-w_l)$ for any $k,l$ 
and some non-zero $\omega\in\C$.
\item (c) $V_k(w)=\omega w_k$ for any $k$ 
and some non-zero $\omega\in\C$.
\item (d) at least one $V_j(w)$ is 
non-zero and for some fixed $k$ 
the equality $w_lV_k(w)=w_kV_l(w)$ 
is satisfied for all indices $l$.
\item (e) at least one $V_j(w)$ is non-zero, S(w)=0 
and furthermore for some fixed $k$ 
the equality $w_lV_k(w)=w_kV_l(w)$ 
is satisfied for at least $n-3$ 
indices $l$ different from $k$.
\end{itemize}
\end{lemma}
\begin{proof}
The previous lemma tells us that (a) and (b) are equivalent.
Obviously (c) implies (b). If (b) is satisfied 
then $p=V_k-\omega w_k$
is independent of $k$ and 
$$p\sigma=\sum_k\Gamma_k(V_k-\omega w_k)
=\sum\Gamma_kV_k-\omega M(w)=\sum\Gamma_kV_k=0.$$
Thus $p=0$. Clearly (c) and (d) are equivalent. Assume (c).
Then
$$\omega S
=\sum_k\Gamma_k\omega w_k\bar{w}_k
=\sum_k\Gamma_kV_k\bar{w}_k=iL=0$$
and (e) follows. Conversly (see \cite{oneil}, Lemma 1.2.7), 
assume $k=1$
and let $D_l=V_1w_l-V_lw_1$ be zero for $l=4..n$.
We have
$$\Gamma_2D_2+\Gamma_3D_3=\sum\Gamma_l(V_1w_l-V_lw_1)=V_1M=0,$$
$$\Gamma_2\bar{w}_2D_2+\Gamma_3\bar{w}_3D_3
=\sum_l\Gamma_lw_l(V_1w_l-V_lw_1)
=V_1S-w_1(iL)=0.$$
This set of equations has non-zero determinant as $w_2\neq w_3$,
and therefore $D_2=D_3=0$.
\end{proof}

\section{Some analytic results on collapsing}
\subsection{Collapsing systems of three vortices}

Collapsing systems of three vortices were first described
by Groebli in 1877.
The configuration space is two-dimensional.
The equations given in the previous section amount to
$$M=\sum\Gamma_kz_k=0,
\quad S=\sum\Gamma_k|z_k|=0.$$
For dimensional reasons we expect 
the solution set to be one-dimensional.
In fact the above equations can be easily transformed 
to one quadratic equation in two variables
which determines a circle in the plane.

\subsection{Collapsing systems of four and five vortices}
The motion of four vortices is no longer integrable.
Nevertheless in 1979 Novikov and Sedov
gave explicit examples
of collapsing systems 
of four and five vortices.
For four vortices their approach works
basically for only one specific circulation
and doesn't give all possible
collapsing systems even for that circulation.
From our point of view the specific circulation
has the property that the polynomial system of equations
factorizes so that one component of the solution set
(in the configuration space) forms a circle (or ellipse)
in a plane contained in the projective configuration space.
These examples can be therefore thought of 
as prototypes of a general situation described by O'Neil.

\subsection{O'Neil's results}
We shall briefly sketch those results obtained by O'Neil in 1987
which are of importance for this paper.
In his setting the existence of collapsing configurations
for arbitrary $n$ follows from two theorems.
\begin{theorem}(O'Neil \cite{oneil} Theorem 7.1.1)
Suppose $\Gamma_l>0$ for $l=1,...,n-1$,
$\Gamma_n<0$, and $L=0$. 
Then there are at least $s(n-2)!$ 
collinear rotational configurations, 
where $s$ is the number of pairs $(l,n)$ 
such that $\Gamma_l+\Gamma_n>0$.
\end{theorem}
\begin{theorem}(O'Neil \cite{oneil} Theorem 7.4.1)
Let $n>3$. For each complex $\Gamma$
satisfying $\Gamma_1=1$ and $L=0$, 
except for a subvariety of codimension 1, 
every collinear rotational configuration
lies on a one-dimensional family of collapsing configurations. 
Each family is a submanifold except at a finite number of points.
\end{theorem}

\section{Approximations of collapsing configurations}
In this section we shall describe an approach
leading to high-accuracy
numerical approximations of collapsing configurations.
We shall illustrate the approach with a specific example 
of a collapsing system of seven vortices.

In the preceeding sections we reviewed some algebraic conditions
for a system of vortices to be collapsing. 
Write such a system of equations as
$$f(z)=0,\quad f:\C^n\to\R^m.$$
We find a numerical solution of this system of equations
using a basically very simple approach:
take an arbitrary point in $\C^n$ and
follow the integral curve of the vector field
$-\nabla g$, $g(z)=|f(z)|^2$,
as the method of steepest decent tells us to do,
until we obtain a local minimum of $g$.
Accept this point as a zero of $f$ 
if the value $f(z)$ is small enough.

Several issues require some care.
First of all it is not clear what range 
for the initial points should be taken.
We were lucky enough to obtain solutions 
for many randomly taken initial points.
Secondly it is not clear how small 
the value of $g(z)$ should be 
in order to be accepted.
Besides the steepest descent method works fine
for a vigorously changing function,
but is much less efficient when $f$ is close to zero,
since the derivatives of $g$ are close to zero then.
We decide to accept a point for which the steepest
descent method gives a value of $g(z)$ of several orders
smaller than the initial value, 
and change the method to one working 
like the Newton's method then.
This enables us to obtain solutions of accuracy 
of several dozen (or even several hundred) decimal digits.
The importance of high accuracy will be made clear later.

Another issue is the unique specification of the solution.
The O'Neil's results suggest that the collapsing configurations
should form one-dimensional familes (smooth curves).
Such a curve projects by a submersion 
onto at least one coordinate axis.
It follows that by fixing a value of this coordinate
we should obtain a system of equations
with discrete solutions.
This means that an initial point close to a solution
should specify that solution as a unique solution 
in some explicitely given neigbourhood of that point.

Finally let's mention an important question 
of whether the approximate solution
we found is really close to a true 
(precise) solution of the system.
This has been taken care of by an approach
using the implicit function theorem.
Without going into details note that to this end
it is enough to bound second derivatives 
of $f$ in a neighbourhood
of the point we found and show that both the value of $f$
and the norm of the inverse of the first 
derivative matrix at the point
is small enough.
Note also that the rigorous proof of the existence of
collapsing systems thus obtained does not rely on O'Neil's results.
In particular our approach works for diverse circulations,
not necessarily meeting 
the conditions given in the first O'Neil's theorem.
Moreover the collapsing configurations we obtain
have no direct relation to any collinear
rotational configuration,
as the second O'Neil's theorem requires.

We shall illustrate the above procedure
for seven vortices 
and for the circulation 
$\Gamma=2\pi\ (2,2,-4,-4,-4,-4,3)$.
We have $L=0$ and $\sigma=-9\cdot2\pi\neq0$,
so that the necessary condition on $\Gamma$
for the existence of collapsing configurations is fulfilled.
We know from one of previous sections 
that collapsing configurations of $n$ vortices
can be represented by solutions 
of the following system of algebraic equations:
$$M(z)=0,
\quad S(z)=0,
\quad V_1z_k=V_kz_1\quad(2\le k\le n-2).$$
For $n=7$ it corresponds to eleven real equations.
These equations are complex homogeneous
and therefore we can fix any non-zero 
complex coordinate of the solution.
By random search we find that the point
$P=(-1.31+10.00i,\ 4.51+5.39i,\ -1.27+7.59i,\ 1.21+2.06i,
\ -0.58-1.00i,\ 2.99-0.95i,\ 1.00+0.00i)$
may be close enough to a solution.
Therefore we choose to complete 
our set of equations with $z_7=1$
and $y_1=Im(z_1)=10$.
Now we have fourteen equations;
we expect this restricts the set of solutions to a finite set.
Thus our aim is to solve the equation $f(z)=0$ for
$$\C^7\ni z\to(V_1z_2-V_2z_1,V_1z_3-V_3z_1,
V_1z_4-V_4z_1,V_1z_5-V_5z_1,$$
$$z_7-1,\sum_k\Gamma_kz_k,\sum_k\Gamma_k|z_k|^2,
Im(z_1)-10)\in\C^6\times\R^2\approx\C^7.$$
By the method of steepest descent we follow the integral curve 
of the vector field $-\nabla g$ for $g(z)=|f(z)|^2$, 
starting from $P$ (given above).
The method allows us to arrive at a point $Q$ 
where the value $g(Q)=1.65\cdot10^{-6}$,
which is smaller than $g(P)=7.29$
by a factor of $2.26\cdot10^7$.
The values $y_1$ and $z_7$ are kept fixed 
and the other coordinates change by less than 0.01.
A measure of quality of the obtained point 
is the value $\omega=(V_k-V_l)/(z_k-z_l)$,
which should be independent of $k,l$ and have negative real part.
At $Q$ the real part of $\omega$ varies 
between  $-0.02625$ and $-0.02595$,
and $Im(\omega)$ is in $[-0.23185,-0.23118]$.

Now we change the method of seekeing a solution to one of Newton's type,
which allows us to obtain a solution practically
with as high accuracy as we wish.
With the standard machine precision of around 18 digits we can
get immediately a point $R$ with $g(R)\le10^{-34}$
and with 17 accurate digits in $Re(\omega)$. 
The expected collapse time is 
$t_\infty=\frac{-1}{2Re\omega}\approx 19.1806.$

\begin{figure}[ht]
\caption{Trajectories of $z_1$, $z_2$ and $z_4$
for time in different ranges.}
\includegraphics[width=140mm]{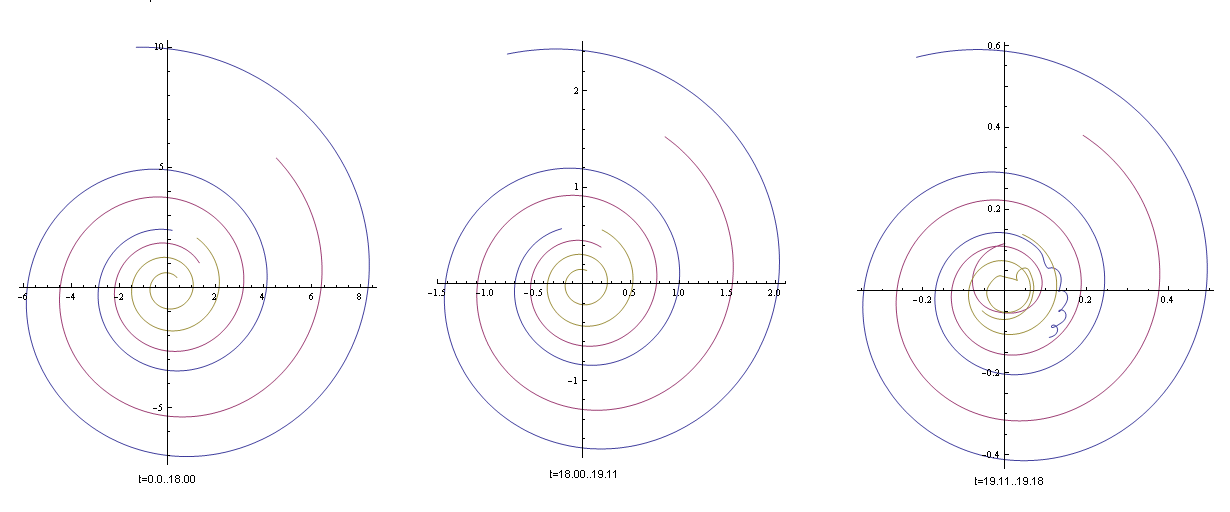}
\end{figure}

The three pictures in Fig. 1 contain the trajectiories 
of three sample vortices $z_1,z_2$ and $z_4$
for time in the range $[0.0, 18.0]$,
$[18.00, 19.11]$ and $[19.11, 19.18]$ respectively. 
Note the scale change in the pictures.
As the distance $|z_1(t)-z_2(t)|$
changes from $7.423$ at $t=0$ to $0.1866$ at $t=19.168$,
the collapsing scale is almost $40$.
The trajectories would seem to truly collapse 
to a point in a single low-resolution picture.
In the third picture suitable magnification
shows that for time close to $19.18$
the self-similarity of the system during evolution is lost.
An attempt to more deeply understand 
what happens near the critical (collapsing) time
is deferred to the next section.
Figure 2 shows the trajectory of $z_1$
in its full time-range. 
We also show the plot of the real and imaginary part of $z_1$,
separately for $t\in[0.0, 19.18]$, 
and separately in the critical range $[19.15, 19.18]$.
Again the loss of self-similarity is clearly seen 
for $t$ close to $19.18$.
\begin{figure}[ht]
\caption{Behavior of $z_1$ and $x_1=Re(z_1)$,
$y_1=Im(z_1)$ for $t\in[0.0, 19.18]$.}
\includegraphics[width=140mm]{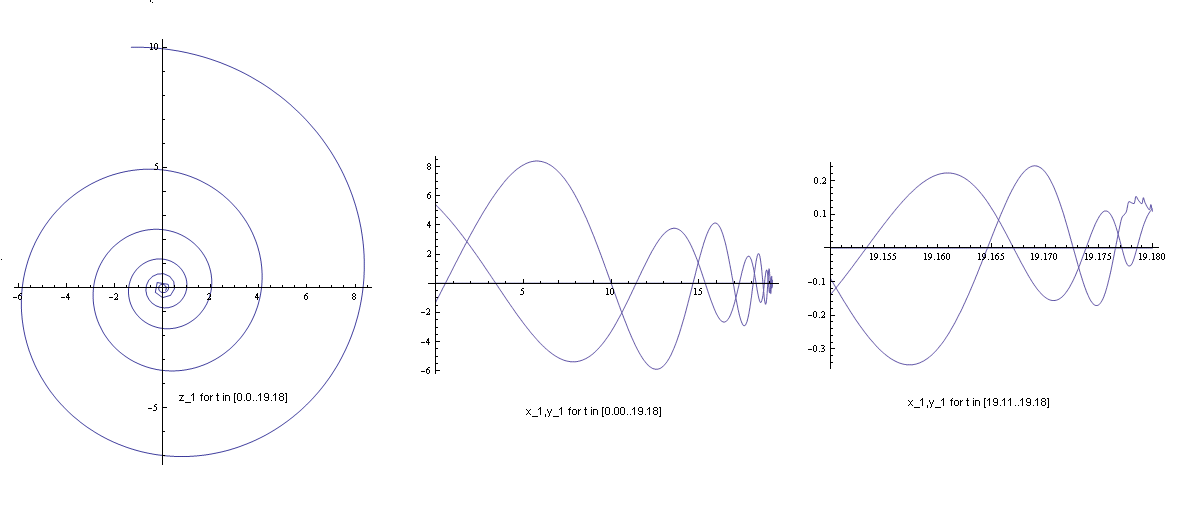}
\end{figure}
Figures 3 to 5 show the streamlines 
of the flow at different moments of time.
Note that the scale changes as the system collapses.
For time in $[0.0, 19.175]$ the system retains self-similarity
- the only visible change is rescaling and rotation.
The state of the system in two last pictures, 
for $t=19.177$ and $t=19.198$,
clearly is not similar to the previous states. 
Moreover in the last picture
the system begins to grow.
Note also that if the scale used for $t=0.0$
was applied to the last state $(t=19.177)$,
it would be impossible to visibly distiguish
the system from a single point vortex.
\begin{figure}[ht]
\caption{The streamlines of the flow
for $t=0.0, 12.0$}
\includegraphics[width=140mm]{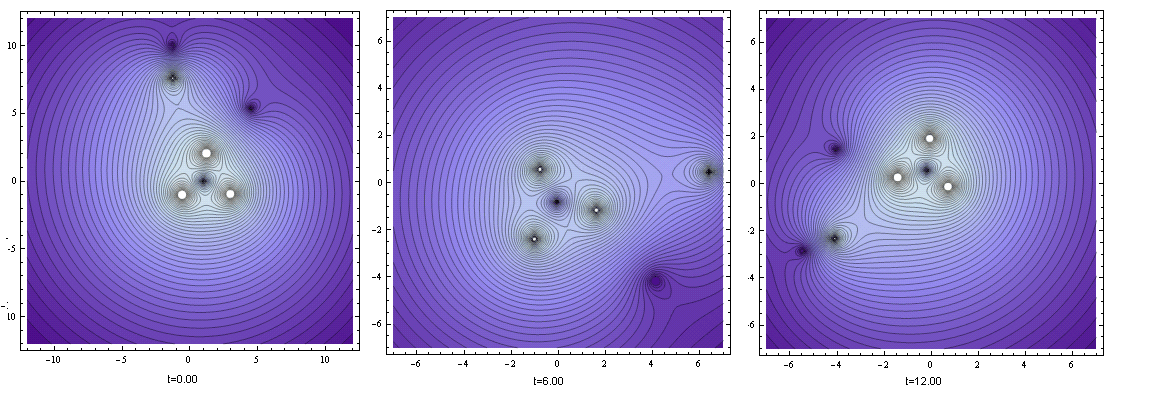}
\end{figure}
\begin{figure}[ht]
\caption{The streamlines of the flow
for $t=19.000, 19.168$}
\includegraphics[width=140mm]{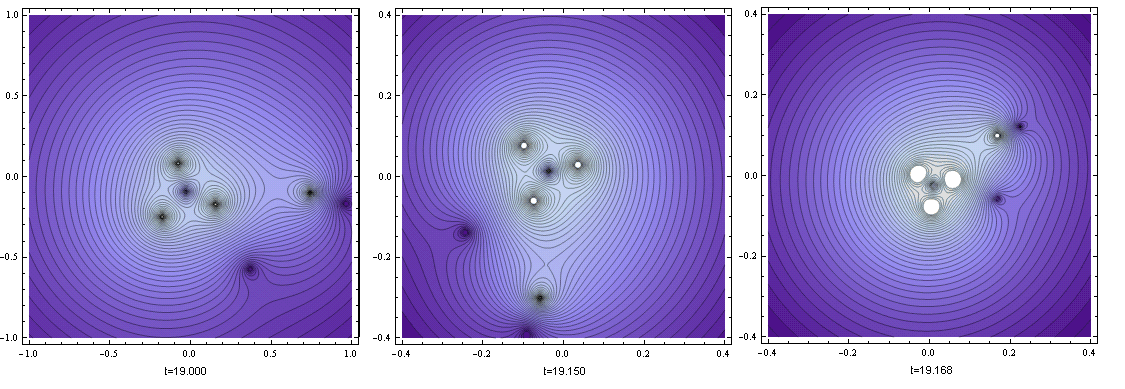}
\end{figure}
\begin{figure}[ht]
\caption{The streamlines of the flow
for $t=19.175, 19.198$}
\includegraphics[width=140mm]{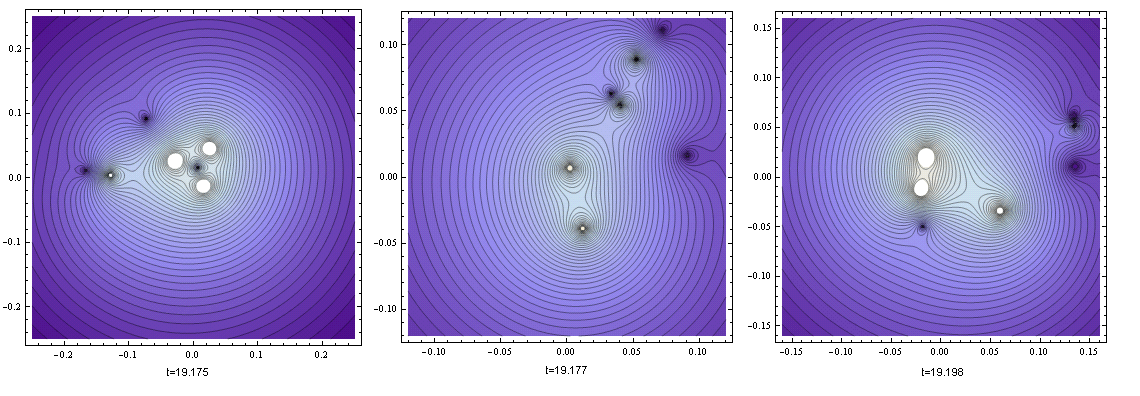}
\end{figure}

\section{Collapsing near critical time}
In this last section we want to address a question
which in our opinion is quite exciting 
and deserves further investigation.

Numerical examples of collapsing configurations
are given with some definite accuracy
and the best one can expect is that their evolution
resembles true collapse only in some
interval $[0,T]$ with $T\approx t_\infty$.
When time nears $t_\infty$
then the system should stay self-similar
with diminishing diameter,
but after passing $T$ the motion
becomes chaotic and the diameter goes up.
It turns out that in order to obtain
a system whose diameter lessens
by a factor of $10^3$
it may be necessary
to know the initial state of the system
with the accuracy of several dozen decimal digits.
This behavior is a feature of the dynamical system 
and not caused by the inaccuracy 
of numerical computations.
We illustate this issue
by taking different approximations 
of the collapsing system studied above
and presenting the behavior of those approximate solutions
during evolution
for times close to the collapsing time $t_\infty$ of the system.
We pay special attention to the collapsing scale,
that is the ratio of the diameter 
of the system at the initial state
and the minimal diameter of the system during evolution.
We want to observe the dependence of the collapsing scale
on the accuracy of the initial state.

Using {\it Mathematica}, we first solved 
the algebraic system of equations
with some precision (the number of decimal digits used)
in the range $[10..200]$,
and then took the so-obtained approximate collapsing system
as the initial data for the differential evolution equation,
which was solved with the same precision.
We tried to find the minimal value of $|z_1(t)-z_2(t)|^2$.
The minimum found and the corresponding value of time
are given in the table below. Tha last column contains
the calculated value of collapsing scale 
(the ratio of the initial diameter and the minimal one).
$$
\begin{array}{crrr}
prec&t&|z_1(t)-z_2(t)|^2&scale\\
10&19.2320&3.4000&30\\
15&19.1080&0.2000&123\\
20&19.1390&0.0100&551\\
30&19.1695&0.0130&483\\
50&19.1791&0.0016&1377\\
80&19.1802&5.5\cdot10^{-5}&7430\\
90&19.1804&9.8\cdot10^{-5}&5566\\
100&19.1804&2.0\cdot10^{-6}&38965\\
120&19.1804&1.0\cdot10^{-6}&55105\\
150&19.1804&2.3\cdot10^{-9}&1149018\\
200&19.1804&2.6\cdot10^{-10}&3417467\\
\end{array}
$$
The log-log graph in Figure 6
shows how the collapsing scale
depends on precision.
\begin{figure}[ht]
\label{logscale}
\caption{The log-log dependency of the collapsing scale on precision}
\includegraphics[width=120mm]{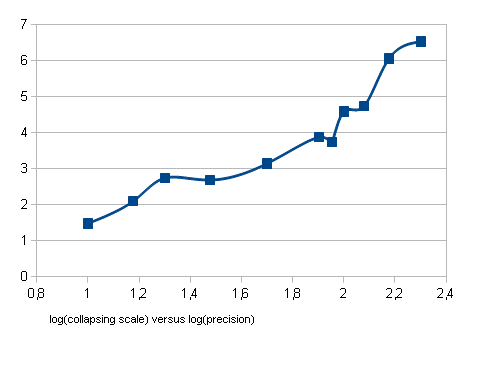}
\end{figure}

We also show some sample plots
of the function $|z_1(t)-z_2(t)|^2$ 
at the places of interest.
Note that for the self-similar collapsing evolution 
the function should be linear in $t$.
This is clearly seen up to a place close to the critical time, 
where self-similarity is lost.
\begin{figure}[ht]
\caption{$|z_1(t)-z_2(t)|^2$ for precision 20}
\includegraphics[width=140mm]{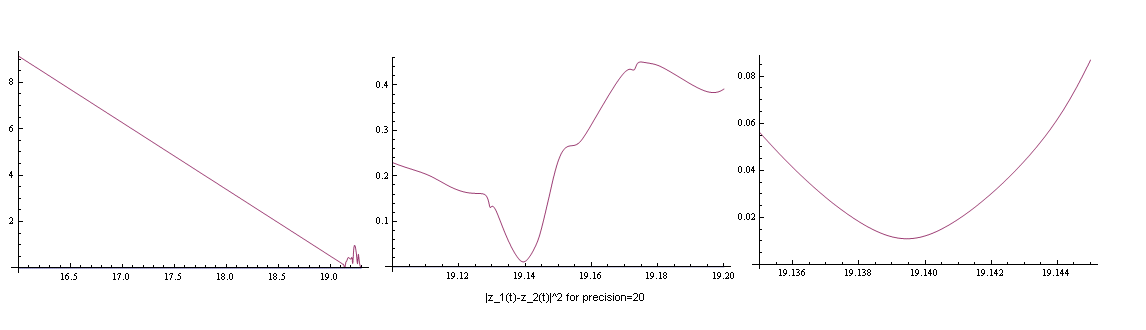}
\end{figure}
\begin{figure}[ht]
\caption{$|z_1(t)-z_2(t)|^2$ for precision 200}
\includegraphics[width=140mm]{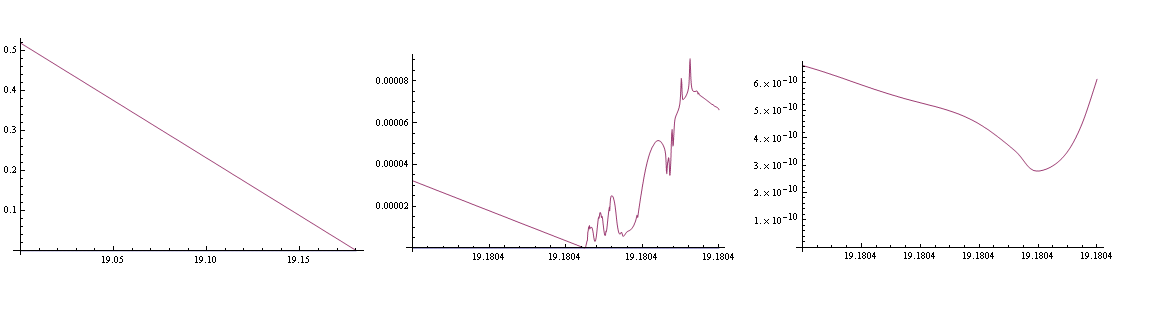}
\end{figure}

\bibliographystyle{plain}	
\bibliography{kkmp}
\end{document}